\documentclass[a4paper,USenglish,cleveref,autoref,numberwithinsect] {lipics-v2021} 

\pdfoutput=1 
\hideLIPIcs  

\title{On first-order model checking parameterized by the number of variables} %

\author{Jan Jedelský}{Faculty of Informatics, Masaryk University, Brno, Czech Republic}{xjedelsk@fi.muni.cz}{https://orcid.org/0000-0001-9585-2553}{}

\authorrunning{J. Jedelský} 

\Copyright{Jan Jedelský} 



\category{} 

\funding{The author has been supported by the project 26-21334S of the Czech Science Foundation.}

\nolinenumbers 

\EventEditors{John Q. Open and Joan R. Access}
\EventNoEds{2}
\EventLongTitle{42nd Conference on Very Important Topics (CVIT 2016)}
\EventShortTitle{CVIT 2016}
\EventAcronym{CVIT}
\EventYear{2016}
\EventDate{December 24--27, 2016}
\EventLocation{Little Whinging, United Kingdom}
\EventLogo{}
\SeriesVolume{42}
\ArticleNo{23}

\usepackage[utf8]{inputenc}
\usepackage[T1]{fontenc}
\usepackage{amsmath}
\usepackage{amssymb}
\usepackage{amsthm}
\usepackage{stmaryrd}
\usepackage{mathrsfs}
\usepackage{mathtools}

\usepackage{url}
\usepackage{todonotes}
\usepackage{mathrsfs}
\usepackage{soul}
\usepackage{relsize}
\usepackage[all,defaultlines=3]{nowidow}
\usepackage[mathlines]{lineno}
\usepackage{multicol}
\usepackage{xspace}
\usepackage{lineno}
\usepackage{nicefrac}
\usepackage{wasysym}
\usepackage{textpos}

\usepackage{etoolbox}
\usepackage{refcount}
\usepackage{apxproof}
\usepackage{keytheorems}[overload]

\newcommand{\NN}{\mathbb{N}}
\renewcommand{\AA}{\mathbb{A}}

\DeclareMathOperator{\tp}{\textsf{tp}}

\DeclareMathOperator{\tw}{tw}

\newkeytheorem{case}[name=Case]
\AtBeginEnvironment{proof}{\setcounter{case}{0}}
\newtheorem{case2}{Case}[case]
\newtheorem{question}[theorem]{Question}

\newkeytheorem{retheorem}[sibling=theorem,name=Theorem]
\newkeytheorem{recorollary}[sibling=corollary,name=Corollary]
\newkeytheorem{reclaim}[sibling=claim,name=Claim]

\def\fixcref#1#2#3#4{
    \AddToHook{env/#1/begin}{\crefalias{theorem}{#1}}
    \crefname{#1}{#1}{#2}
    \Crefname{#1}{#3}{#4}
    \crefformat{#1}{##2#1~##1##3}
    \Crefformat{#1}{##2#3~##1##3}
}
\fixcref{theorem}{theorems}{Theorem}{Theorems}
\fixcref{corollary}{corollaries}{Corollary}{corollaries}
\fixcref{lemma}{lemmas}{Lemma}{Lemmas}
\fixcref{question}{questions}{Question}{Questions}
\fixcref{conjecture}{conjectures}{Conjecture}{Conjectures}
\fixcref{claim}{claims}{Claim}{Claims}

\newcommand{\classE}{$\mathsf{E}$}
\newcommand{\classNE}{$\mathsf{NE}$}

\newcommand{\FPT}{$\mathsf{FPT}$}
\newcommand{\XP}{$\mathsf{XP}$}

\newcommand{\AW}[1]{$\mathsf{AW[#1]}$}

\newcommand{\ca}{\mathcal}

\newenvironment{subproof}[1][Subproof]{%
  \renewcommand{\qedsymbol}{$\blacksquare$}%
  \begin{proof}[#1]%
}{%
  \end{proof}%
}

\usepackage{todonotes}

\begin{document}

\maketitle

\begin{abstract}
    The first-order (FO) model checking problem asks, given an FO sentence $\phi$ and a graph $G$, whether $G$ is a model of $\phi$. This problem is $\mathsf{AW[*]}$-hard when parameterized by the quantifier rank of the formula. A classical algorithm decides this problem in $\mathsf{XP}$ time when parameterized by the number of variable names in the formula.

    This $\mathsf{AW[*]}$-hardness motivates study of the problem on well-behaved graph classes. Many results identify graph classes $\mathcal{C}$ such that the FO model checking problem restricted to $\mathcal{C}$ admits an $\mathsf{FPT}$ time algorithm when parameterized by the quantifier rank of the formula.

    Parameterization by the number of variable names is significantly more restrictive than parameterization by the quantifier rank, as the number of variable names can be arbitrarily small relative to the quantifier rank.
    We study the graph classes $\mathcal{C}$ for which the FO model checking problem restricted to $\mathcal{C}$ admits an $\mathsf{FPT}$ time algorithm when parameterized by the number of variable names in the formula. We characterize these classes in the monotone setting, and prove a slightly weaker result in the hereditary setting.
\end{abstract}

\newpage

\section{Introduction}

The first-order (FO) model checking problem asks, given a graph $G$ and an FO sentence $\phi$, whether $G \models \phi$. The FO model checking problem thus generalizes all problems expressible in first-order logic. The parameterization of the model checking problem usually depends on the formula. There are two very natural parameters of a formula -- its quantifier rank, and the number of variables (that is, the number of variable names). Given an $n$-vertex graph $G$ and an FO sentence $\phi$ using at most $k$ distinct variables, the standard naive evaluation algorithm decides whether $G \models \phi$ in time $\ca O(|\phi| \cdot n^k)$. That is, the FO model checking problem parameterized by the number of variables can be solved in \XP-time. Downey, Fellows, and Taylor \cite{DBLP:conf/dmtcs/DowneyFT96} proved that FO model checking parameterized by the quantifier rank of the sentence is \AW{*}-hard, which motivates the following question:
\begin{question}\label{q:fo-mc-param-by-qr}
    For which graph classes $\ca C$ does FO model checking admit an \FPT-time algorithm when restricted to inputs from $\ca C$ and parameterized by the quantifier rank of the formula?
\end{question}

Over the last few decades, there has been significant progress toward \Cref{q:fo-mc-param-by-qr}. 
Grohe et al. \cite{DBLP:journals/jacm/GroheKS17} proved that FO model checking parameterized by quantifier rank is (non-uniformly) \FPT{} on nowhere dense graph classes,
and Kreutzer \cite{DBLP:books/cu/11/Kreutzer11} proved hardness of the same problem
on any monotone graph class which is not nowhere dense.

In the hereditary setting, there has been rapid progress in recent years.
FO model checking parameterized by quantifier rank has been shown to be \FPT{} on classes of bounded twin-width \cite{DBLP:journals/jacm/BonnetKTW22} (assuming a witness of small twin-width is provided), structurally nowhere dense classes \cite{DBLP:conf/stoc/DreierMS23},
and monadically stable classes \cite{DBLP:conf/focs/DreierEMMPT24}.
Finally, Dreier and Toruńczyk \cite{DBLP:conf/stoc/DreierT25} introduced merge-width, a parameter generalizing twin-width and bounded expansion, and gave an algorithm for the FO model checking problem on classes of bounded merge-width (again, assuming a witness of small merge-width for a suitable radius is provided). The existence of an algorithm for all hereditary monadically dependent classes remains open. For ordered graphs, Bonnet et al. \cite{DBLP:journals/jacm/BonnetGMSTT24} proved that bounded twin-width coincides with monadic dependence.
On the hardness side, Dreier, M{\"a}hlmann, and Toru\'nczyk \cite{DBLP:conf/stoc/DreierMT24} proved that the problem is \AW{*}-hard on hereditary monadically independent classes.

\medskip

The naive evaluation algorithm runs in \XP-time even when parameterized by the number of variables rather than the quantifier rank of the formula. This motivates a second question:
\begin{question}\label{q:fo-mc-param-by-num-var}
    For which graph classes $\ca C$ does FO model checking admit an \FPT-time algorithm when restricted to inputs from $\ca C$ and parameterized by the number of variables in the formula?
\end{question}

We investigate \Cref{q:fo-mc-param-by-num-var} for monotone and hereditary classes $\ca C$.
We give a kernel for shallow colored trees, and we show that the problem is hard already on paths:

\getkeytheorem{thm:tree-algo}
\getkeytheorem{thm:paths-hard}

Using Theorems~\ref{thm:tree-algo}~and~\ref{thm:paths-hard}, we fully answer the question for monotone graph classes:

\getkeytheorem{col:algo-tree-depth}
\getkeytheorem{thm:unbounded-tree-depth-hardness}

\Cref{thm:unbounded-tree-depth-hardness} together with \Cref{col:algo-tree-depth} immediately gives the following solution to \Cref{q:fo-mc-param-by-num-var} in the monotone setting:

\begin{retheorem}[store=thm:tree-depth-is-the-limit,label=thm:tree-depth-is-the-limit]
    Assume that \FPT$\neq$\AW{*}. Let $\ca C$ be a monotone class of graphs. Then, FO model checking is \FPT when restricted to inputs from $\ca C$ and parameterized by the number of variables if and only if $\ca C$ has bounded tree-depth.
\end{retheorem}

We conjecture that an analogue of \Cref{thm:tree-depth-is-the-limit} also holds for hereditary classes, with tree-depth replaced by shrub-depth. We use Theorems~\ref{thm:tree-algo}~and~\ref{thm:paths-hard} to obtain an algorithm and a conditional lower bound supporting our conjecture:

\begin{corollary}[\Cref{col:algo-shrub-depth}]
    For every graph class $\ca C$ of bounded shrub-depth, there exist a computable function $f_{\ca C}: \NN \to \NN$ and an algorithm $\AA$ such that, for every integer $s \in \NN$, given an FO sentence $\phi$ with $s$ variables and a graph $G \in \ca C$, $\AA$ decides whether $G \models \phi$ in time $\ca O(f_{\ca C}(s) \cdot (|G|^{\ca O(1)}+|\phi|))$.
\end{corollary}

\getkeytheorem{thm:unbounded-shrub-depth-hardness}

We defer the definition of a flipped half-graph and a layerwise flipped $tP_t$ to \Cref{def:half-graph-paths-flips}. Although we do not know how to compute the flip in Item~(\ref{item:algo-flip-of-tPt}), such a flip exists whenever Item~(\ref{item:has-flipped-half-graphs}) does not hold. Indeed, M{\"{a}}hlmann \cite{DBLP:conf/icalp/Mahlmann25} proved that every hereditary class $\ca C$ of unbounded shrub-depth contains a flipped half-graph of order $t$ for all $t \ge 1$ or $\ca C$ contains a layerwise flipped $tP_t$ for all $t \ge 1$.

\medskip

In \Cref{sec:app}, we use our model-checking algorithms to derive an \FPT{} algorithm for deciding the existence of a homomorphism from $H$ to $G$ parameterized by the tree-width of $H$ and the SC-depth of $G$. We also use \Cref{thm:reduce-graph} to characterize classes of bounded shrub-depth in purely model-theoretic terms.

\getkeytheorem{cor:shrub-depth-unbounded-num-types}

The kernelization of FO model checking on shallow trees parameterized by the number of variables raises the analogous question for MSO. We show that MSO model checking with this parameterization admits no kernel on shallow trees unless \FPT$=$\AW{*}:

\getkeytheorem{thm:MSO-aw-hard}

\subsection*{Outline of the paper}

In \Cref{sec:prelim}, we review the notation and the relevant concepts from graph theory and finite model theory. 
We present the algorithms in \Cref{sec:algo} and prove the hardness results in \Cref{sec:hardness}.
In \Cref{sec:app}, we give two applications of our results. Finally, we discuss conclusions and open questions in \Cref{sec:concl}.

\section{Preliminaries}\label{sec:prelim}

We briefly recall basic notation. Given a graph $G$, we denote by $V(G)$ its vertex set and by $E(G)$ its edge set. We also denote by $N_G(u)$ the open neighborhood of a vertex $u \in V(G)$, and by $N_G[u]$ the closed neighborhood.
Finally, we denote by $G[X]$ the subgraph of $G$ induced by a set $X \subseteq V(G)$.
The height of a rooted tree $T$ is the number of vertices on its longest leaf-to-root path, and the depth of $T$ is the length of that path.
Let $c \in \NN$ be a natural number. A $c$-colored graph has each vertex colored by exactly one of $c$ colors; we usually take the set of colors to be $[c]:=\{1,2,\ldots,c\}$. An isomorphism of colored graphs must preserve the colors.
We only consider finite nonempty simple graphs, but most of our results could be extended to other finite relational structures.

We say that a function $f(\bar a, \bar b): \NN^\alpha \times \NN^\beta \to \NN$ is \emph{elementary for any fixed $\bar b$} if, for every $\bar b \in \NN^\beta$, the function $g_{\bar b}(\bar a):=f(\bar a, \bar b)$ is elementary.

\subsection{Half-graphs, paths, and flips}\label{def:half-graph-paths-flips}

Let $G$ be a graph and let $\ca P$ be a partition of a subset $\bigcup \ca P$ of the vertex set of $G$. A \emph{$\ca P$-flip} is a symmetric relation $F$ on $\ca P$. We denote by $F \oplus G$ the following graph: $V(F \oplus G) = V(G)$, and $uv \in E(F \oplus G)$ where $u \neq v$ is decided as follows:
\begin{itemize}
    \item If $\{u,v\} \not\subseteq \bigcup \ca P$, then $uv \in E(F \oplus G)$ if and only if $uv \in E(G)$.
    \item Otherwise, there are $U,V \in \ca P$ such that $u \in U$ and $v \in V$. Then, $uv \in E(F \oplus G)$ if and only if ($uv \in E(G)$ xor $(U,V)\in F$).
\end{itemize}
If $H$ is an induced subgraph of $G$, we also write $F \oplus H := (F \oplus G)[V(H)]$ to denote the subgraph of $F \oplus G$ induced by $V(H)$.
In a slight abuse of notation, we sometimes say that a graph $K$ is a \emph{$\ca P$-flip} of $G$ if there is a $\ca P$-flip $F$ such that $K=F \oplus G$. If $\ca P=\{P\}$ has only one part, then we write $P$-flip instead of $\{P\}$-flip.

A \emph{half-graph} $H_t$ of order $t$ is a bipartite graph on $2t$ vertices $a_1, \ldots, a_t$ and $b_1, \ldots, b_t$ such that $a_ib_j \in E(H_t)$ if and only if $i \le j$. We always assume that the bipartition is $A_t=\{a_1, \ldots, a_t\}$ and $B_t=\{b_1, \ldots, b_t\}$. A \emph{flipped half-graph} of order $t$ is any $\{A_t, B_t\}$-flip of $H_t$.

We denote by $kP_t$ the disjoint union of $k$ paths on $t$ vertices. We assume that $V(kP_t) = [k] \times [t]$, and $(a,b)(c,d) \in E(kP_t)$ if and only if $a=c$ and $|b-d|=1$. With each $kP_t$, we associate a layering $\ca L_{k,t}=(L_1, \ldots, L_t)$, where $L_i=\{(a, i) \mid a \in [k]\}$. A \emph{layerwise flipped $kP_t$} is any $\ca L_{k,t}$-flip of $kP_t$.

\subsection{Graph parameters}

\begin{definition}[Elimination forests and tree-depth]
Let $G$ be a graph, and let $T$ be a rooted forest on the same vertex set. We say that $T$ is an \emph{elimination forest} of $G$ if the endpoints $u,v$ of every edge $uv$ of $G$ are in an ancestor-descendant relation in $T$. The \emph{tree-depth} of $G$ is the minimum height of its elimination forest.
\end{definition}

\begin{definition}[Tree-models and shrub-depth]
Let $G$ be a graph, and let $T$ be a $c$-colored rooted tree such that the set of leaves of $T$ is exactly the vertex set of $G$, and all leaves are at the same distance from the root. We say that $T$ is a \emph{tree-model} of $G$ if the adjacency in $G$ of two vertices $u,v \in V(G)$ depends only on their colors and their distance in $T$.

A graph class $\ca C$ has \emph{shrub-depth} at most $t$ if there is $c \in \NN$ such that every graph $G \in \ca C$ admits a tree-model of depth at most $t$ using at most $c$ colors.
\end{definition}

Note that there is no notion of ``shrub-depth of a graph''. Instead, one uses a related parameter that combines both the number of colors and the height of a tree-model into one number:

\begin{definition}[SC-depth]
    The class of graphs $\ca S \ca C_t$ with \emph{SC-depth} at most $t$ is inductively defined as follows:
    \begin{itemize}
        \item $\ca S \ca C_0 = \{K_1\}$ consists of the unique single-vertex graph, up to isomorphism.
        \item Let $G_1, \ldots, G_\ell \in \ca S \ca C_i$, let $H = G_1 \dot\cup \ldots \dot\cup G_\ell$ be their disjoint union, and let $X \subseteq V(H)$ be any subset of vertices of $H$. Denote by $\overline{H}^X$ the $X$-flip obtained from $H$ by flipping adjacency within $X$. Then, $\overline{H}^X \in \ca S \ca C_{i+1}$.
    \end{itemize}
    The \emph{SC-depth} of a graph $G$ is the least integer $t$ such that $G \in \ca S \ca C_t$.
\end{definition}

It is well known that a graph class $\ca C$ has bounded shrub-depth if and only if $\ca C$ has bounded SC-depth. Furthermore, if $\ca C$ is weakly sparse, that is, $\ca C$ excludes some biclique as a subgraph, then $\ca C$ has bounded shrub-depth if and only if $\ca C$ has bounded tree-depth.

For hereditary classes of bounded shrub-depth and monotone classes of bounded tree-depth, Gajarsk\'y and Hlin\v{e}n\'y \cite{DBLP:journals/corr/abs-1204-5194} proved that MSO model checking is \FPT{} when parameterized by formula size, with elementary dependence on this parameter. They prove this result by kernelization, and we use a similar approach to prove \Cref{thm:tree-algo}.
On the other hand, Lampis \cite{DBLP:journals/corr/abs-1302-4266}
showed that MSO model checking cannot be decided on paths with elementary dependence on the size of the formula, unless \classE$=$\classNE{}.
Hence, elementary dependence is also ruled out for monotone classes of unbounded tree-depth.
M{\"{a}}hlmann \cite{DBLP:conf/icalp/Mahlmann25} described obstructions that allow one to interpret the class of all paths in any hereditary class of unbounded shrub-depth. If those obstructions can be efficiently computed, then elementary dependence is also ruled out for hereditary classes of unbounded shrub-depth.

\medskip

\begin{definition}[Tree-decomposition and tree-width]
    Let $G$ be a graph, let $T$ be a tree, and let $\ca X = (X_t)_{t \in V(T)}$ be a collection of bags $X_t \subseteq V(G)$. Assume that each vertex and each edge of $G$ is present in some bag, and for each vertex $v \in V(G)$, the subgraph of $T$ induced by nodes $t \in V(T)$ of bags $X_t \owns v$ containing $v$ is connected.

    Then, $(T, \ca X)$ is a \emph{tree-decomposition} of $G$, and its width is $\max\limits_{t \in V(T)} |X_t|-1$. The \emph{tree-width} $\tw(G)$ of $G$ is the minimum width over all tree-decompositions of $G$.
\end{definition}

\subsection{First-order logic and pebble games}

We assume that the reader is familiar with FO logic. We only consider finite signatures of colored graphs. Briefly, an FO formula in the language of $c$-colored graphs can existentially ($\exists x: \phi$) and universally ($\forall x: \phi$) quantify over vertices of the graph; it can use the usual Boolean connectives ($\phi \land \psi$, $\phi \lor \psi$, $\lnot \phi$, $\ldots$); it can express equality ($x = y$) and adjacency ($E(x, y)$) of a pair of variables $x, y$; and it can express that the vertex assigned to $x$ has color $i \in [c]$ ($C_i(x)$). A \emph{sentence} is a formula without any free variables. The \emph{quantifier rank} of a formula is the maximum depth of quantifier nesting. For a fixed finite signature and quantifier rank, there are only finitely many sentences up to logical equivalence.
The \emph{number of variables} of a formula is the number of distinct variable names. Variable names can be reused. One can normalize any sentence $\phi$ (by renaming variables while preserving logical equivalence) so that the quantifier rank of $\phi$ is greater than or equal to the number of variables of $\phi$. On the other hand, it is easy to see that three variables suffice to write a formula $\phi_k(x,y)$ expressing ``there is a path of length at most $k$ between $x$ and $y$''. Hence, there are infinitely many logical equivalence classes of sentences with three variables, so one cannot bound the quantifier rank of a sentence $\phi$ in terms of the number of variables of $\phi$. We denote by FO$^s$ the class of all first-order formulas with at most $s$ variables.
Given a sentence $\phi$ and a (colored) graph $G$, we say that $G$ is a \emph{model} of $\phi$, written $G \models \phi$, if the sentence $\phi$ evaluates to true on $G$.
The \emph{type} $\tp(G)$ of a graph $G$ is the set of all sentences $\phi$ such that $G \models \phi$. An \emph{FO$^s$ type} of $G$ is defined as $\tp^s(G):=\tp(G)\cap$FO$^s$.

We say that two graphs $A$ and $B$ are \emph{elementarily equivalent} if $\tp(A)=\tp(B)$, and we denote this by $A \equiv B$. Similarly, we say that $A$ and $B$ are \emph{FO$^s$ equivalent} if $\tp^s(A)=\tp^s(B)$, and we denote this by $A \equiv^{\text{FO}^s} B$.
To prove elementary equivalence or FO$^s$ equivalence, one uses pebble games. The well-known Ehrenfeucht-Fraïssé game is used to prove elementary equivalence. The game is played by two players, Duplicator and Spoiler, on a pair of graphs $A$ and $B$. Before the first turn, Spoiler announces the number of turns $q$. In each turn, Spoiler picks one of the graphs and places a new pebble on a vertex of that graph, or moves an existing pebble from one vertex to another. Then, Duplicator places a new pebble or moves an existing pebble in the other graph. Duplicator must place a new pebble if Spoiler placed a new pebble, and must move the $i$-th pebble if Spoiler moved the $i$-th pebble. Consider $k$-tuples $\bar a = (a_1, \ldots, a_k)$ of vertices of $A$ and $\bar b=(b_1,\ldots,b_k)$ of vertices of $B$; we say that $\bar a \mapsto \bar b$ is a \emph{partial isomorphism} if the function $f(a_i)=b_i$ is an isomorphism of $A[\{a_1, \ldots, a_k\}]$ and $B[\{b_1, \ldots, b_k\}]$.
Let $\bar a$ be the tuple of vertices from $A$ where $a_i$ is the vertex carrying the $i$-th pebble, and let $\bar b$ be the analogous tuple from $B$. If $\bar a \mapsto \bar b$ is not a partial isomorphism at the end of any of the $q$ turns, then Spoiler wins the Ehrenfeucht-Fraïssé game. Duplicator wins if Spoiler does not win after any round. The graphs $A$ and $B$ are elementarily equivalent if and only if Duplicator has a winning strategy in the $q$-round Ehrenfeucht-Fraïssé game on $A$ and $B$ for every $q$.

In order to prove FO$^s$ equivalence, one uses a variant of the Ehrenfeucht-Fraïssé game, in which the players can only use $s$ pebbles on each graph. Since the players can move the pebbles, this does not limit the number of rounds. We refer to this game as the \emph{FO$^s$ pebble game}. Again, two graphs $A$ and $B$ are FO$^s$ equivalent if and only if Duplicator wins the FO$^s$ pebble game on $A$ and $B$.

We introduce some notation that will be useful for discussing the FO$^s$ pebble game.
Let $G$ be a graph, and let $\bar v=(v_1, \ldots, v_s)$ be an $s$-tuple of elements of $V(G) \cup \{*\}$, where $* \not\in V(G)$. The tuple $\bar v$ records the positions of the $s$ pebbles. If $v_i \neq *$, then the $i$-th pebble is placed on a vertex $v_i$;  otherwise, the pebble has not been placed anywhere on the graph. If $i \in [s]$ is an index and $u$ is a vertex, we denote by $\bar v[i \to u]:=(v_1, \ldots, v_{i-1}, u, v_{i+1}, \ldots, v_s)$ the tuple obtained by moving the $i$-th pebble to $u$. If $H$ is an induced subgraph of $G$, we denote by $\bar v[H]$ the tuple obtained from $\bar v$ by replacing all $v_i \not\in V(H)$ with $*$. We write $\Game(A, \bar a, B, \bar b)$ to denote the state of the FO$^s$ pebble game played on graphs $A$ and $B$, where the $s$-tuple $\bar a \in (V(A) \cup \{*\})^s$ describes the pebbles placed on $A$, and similarly $\bar b \in (V(B) \cup \{*\})^s$  describes the pebbles placed on $B$. We often view $\Game(A, \bar a, B, \bar b)$ as a game rather than a state of the FO$^s$ pebble game. We say that a mapping $\bar a \mapsto \bar b$ is an \emph{$s$-partial isomorphism} of $A$ and $B$ if both $\bar a$ and $\bar b$ have $*$ at the same indices, and the restriction of $\bar a$ and $\bar b$ to non-$*$ elements is a partial isomorphism of $A$ and $B$. Observe that Duplicator has a winning strategy for $\Game(A, \bar a, B, \bar b)$ if and only if $\bar a \mapsto \bar b$ is an $s$-partial isomorphism, and for every move of Spoiler, Duplicator's response results in a game $\Game(A, \bar a', B, \bar b')$ such that $\bar a' \mapsto \bar b'$ is an $s$-partial isomorphism and Duplicator has a winning strategy for $\Game(A, \bar a', B, \bar b')$. Note that this observation uses the finiteness of both $A$ and $B$.

Let $\phi(x)$ be a formula with one free variable, and let $\psi(x,y)$ be a formula with two free variables. Assume that $\psi$ is symmetric and irreflexive, that is, $\psi(x,y)$ is logically equivalent to both $\psi(x, y)\land\psi(y,x)$ and $\psi(x, y)\land x \neq y$. The pair $\iota:=(\phi, \psi)$ is called a \emph{1-dimensional FO interpretation}. All the interpretations in this work are 1-dimensional FO interpretations, so we usually omit these prefixes. Let $G$ be a graph. An \emph{$\iota$-interpretation} of $G$ is a graph $\iota(G)$ such that $V(\iota(G)):=\{v \in V(G) : G \models \phi(v)\}$, and for any pair $u,v \in V(\iota(G))$ we have $uv \in E(\iota(G))$ if and only if $G \models \psi(u, v)$.
The following lemma gives a useful property of FO interpretations:

\begin{lemma}[Backwards Translation Lemma]
    Let $\iota$ be an FO interpretation and $\phi$ be an FO sentence. Then, there is an FO sentence $\phi[\iota]$ such that for every graph $G$, $G \models \phi[\iota]$ if and only if $\iota(G) \models \phi$.
\end{lemma}
We remark that, if $\phi$ is a sentence of FO$^s$, then $\phi[\iota]$ is a sentence of FO$^{s+c}$ for some constant $c$ depending only on $\iota$.

We refer the reader to the book by Ebbinghaus and Flum \cite{finite-model-theory-book} for a more in-depth explanation of the above concepts.

\section{Algorithms}\label{sec:algo}

First, we observe that the classical result of Gajarsk\'y and Hlin\v{e}n\'y \cite{DBLP:journals/corr/abs-1204-5194} on kernels of colored trees of bounded depth extends to FO$^s$ without any bound on quantifier rank:

\begin{lemma}\label{thm:reduce-graph}
    There are computable functions $g, h: \NN \times \NN \times \NN \to \NN$ such that, for every $s \in \NN$, $k \in \NN$, $c \in \NN$, and every rooted $c$-colored tree $T$ of depth $k$, there is a subtree $L$ of $T$ on at most $g(s, k, c)$ vertices such that $T$ and $L$ are FO$^s$-equivalent. Furthermore, such a tree $L$ can be computed from $T$ in time $\ca O(h(s, k, c) \cdot n)$, and both functions $g$ and $h$ are elementary for any fixed $k$.
\end{lemma}
\begin{proof}
    Let $n:=|V(T)|$. The case $s=0$ is immediate: take any one-vertex subtree $L$ and set $g(0,k,c)=h(0,k,c):=1$. Hence, assume that $s\ge1$. We prove the assertion for trees of depth at most $k$ by induction on $k$. If $k=0$, then $T$ has only one vertex, so it suffices to set $g(s, 0, c)=h(s,0,c):=1$ and $L := T$. Hence, we assume that $k \ge 1$.
    We split the proof into two parts. First, we inductively define the subtree $L$ and implicitly give a recursive algorithm for computing it. Second, we show that our construction of $L$ makes $L$ and $T$ FO$^s$-equivalent.

    \medskip

    \noindent\textbf{Construction.}\,
    Let $r$ be the root of $T$. Without loss of generality, we assume that the vertices of $T$ have colors from the set $[c]=\{1,2,\ldots,c\}$, and we denote by $c_v^T$ the color of a vertex $v$ in $T$. Let $T^r$ be the $2c+1$-colored tree obtained from $T$ by giving $r$ a unique color $2c+1$, assigning each neighbor $v \in N_T(r)$ of $r$ the color $c+c_v^T$, and keeping the colors of all the other vertices. Let $T_1^r, \ldots, T_\ell^r$ be the inclusion-maximal subtrees of $T^r-r$. For each branch $T_i^r$, we denote its root by $r_i$, which is the unique neighbor of $r$ in $T_i^r$.

    By the induction hypothesis, for each index $i$, there is a $(c+1)$-colored subtree $L^r_i$ of $T^r_i$ on at most $g(s, k-1, c+1)$ vertices such that $T^r_i$ and $L^r_i$ are FO$^s$-equivalent. In particular, $r_i$ belongs to $L_i^r$ because of its unique color. By Cayley's formula \cite{Cayley_2009}, there are, up to isomorphism, at most $p:=\left(2c \cdot g(s, k-1, c+1)\right)^{g(s, k-1, c+1)}$ distinct colored trees with colors from the set $[2c]$ on at most $g(s, k-1, c+1)$ vertices. We set $g(s, k, c):=1 + s \cdot p \cdot g(s, k-1, c+1)$, and we observe that $g$ is elementary for any fixed $k$. Let $I$ be the set of indices $i\in[\ell]$ of subtrees $T_i^r$ such that there are fewer than $s$ distinct indices $j < i$ for which $L^r_j$ is isomorphic to $L^r_i$. We construct $L$ by restricting $T$ to $\{r\}\cup\bigcup_{i \in I} V(L^r_i)$, which satisfies $|V(L)|\le g(s, k, c)$.

    It suffices to test each tree $L^r_i$ for isomorphism with at most $s \cdot p$ trees $L^r_j$ to compute $I$. Each such test takes $\ca O(c+g(s,k-1,c+1))$ time by the standard tree-isomorphism algorithm. Since $\sum_{i=1}^\ell |V(T_i^r)|=n-1$ and $\ell\le n-1$, $L$ can be constructed from $T$ in time
    \[
        \ca O\left(n+h(s,k-1,c+1)\cdot(n-1)+\ell\cdot s\cdot p\cdot(c+g(s,k-1,c+1))\right)
        \le \ca O(h(s,k,c)\cdot n),
    \]
    where we set $h(s,k,c):=1+h(s,k-1,c+1)+s\cdot p\cdot(c+g(s,k-1,c+1))$. Thus $g$ and $h$ are computable and elementary for any fixed $k$.

    \medskip

    \noindent\textbf{FO$^s$-equivalence.}\,
    Let $L^r$ be the subtree of $T^r$ induced by $V(L)$. We call $T_i^r$ ($i\in[\ell]$) and $L_i^r$ ($i\in I$) the \emph{branches} of $T^r$ and $L^r$, respectively.
    We prove that Duplicator has a winning strategy in $G_0:=\Game(T^r, \underbrace{*\ldots*}_s, L^r, \underbrace{*\ldots*}_s)$, and thus $T^r$ and $L^r$ are FO$^s$-equivalent. This is sufficient to prove the FO$^s$-equivalence of $T$ and $L$, because the original colors in $T$ and $L$ are obtained by the same Boolean combinations of the colors in $T^r$ and $L^r$.
    Let $W$ be the set of games defined as follows:
    \begin{itemize}
        \item Let $\bar t \in \left(V(T^r) \cup \{*\}\right)^s$ and $\bar l \in \left(V(L^r) \cup \{*\}\right)^s$ be $s$-tuples where $\bar t \mapsto \bar l$ is an $s$-partial isomorphism from $T^r$ to $L^r$. We let $\Game(T^r, \bar t, L^r, \bar l) \in W$ if and only if there is an injective partial function $m_{\bar t, \bar l}:I\rightharpoonup[\ell]$ such that the following holds for all indices $i_1,i_2\in[s]$, $j_T\in[\ell]$, and $j_L\in I$:
        \begin{enumerate}
            \item[(i)] If $\{t_{i_1}, t_{i_2}\} \subseteq V(T^r_{j_T})$ and $l_{i_1} \in V(L^r_{j_L})$, then $l_{i_2} \in V(L^r_{j_L})$, $m_{\bar t, \bar l}(j_L)=j_T$, and the colored trees $L^r_{j_T}$ and $L^r_{j_L}$ are isomorphic.
            \item[(ii)] If $\{l_{i_1}, l_{i_2}\} \subseteq V(L^r_{j_L})$ and $t_{i_1} \in V(T^r_{j_T})$, then $t_{i_2} \in V(T^r_{j_T})$, $m_{\bar t, \bar l}(j_L)=j_T$, and the colored trees $L^r_{j_T}$ and $L^r_{j_L}$ are isomorphic.
            \item[(iii)] Denote by $\bar t[T^r_{j_T}]$ (resp. $\bar l[L^r_{j_L}]$) the restriction of $\bar t$ (resp. $\bar l$) to $V(T^r_{j_T})$ (resp. $V(L^r_{j_L})$). If $m_{\bar t, \bar l}(j_L)$ is defined and equals $j_T$, then Duplicator has a winning strategy for $\Game(T^r_{j_T}, \bar t[T^r_{j_T}], L^r_{j_L}, \bar l[L^r_{j_L}])$.
        \end{enumerate}
    \end{itemize}
    If (i) and (ii) are satisfiable, then there is a unique function $m_{\bar t, \bar l}$ with inclusion-minimal domain satisfying (i) and (ii). Its domain consists exactly of the indices $i$ of the branches $L_i^r$ of $L^r$ containing entries of $\bar l$, and its image consists exactly of the indices $i$ of the branches $T_i^r$ of $T^r$ containing entries of $\bar t$. We always choose this function.

    It suffices to show that, from any game in $W$, Duplicator can respond to every move of Spoiler so that the resulting game is again in $W$. Since $G_0\in W$, this yields a winning strategy for Duplicator in $G_0$.

    Fix an arbitrary game $G_1:=\Game(T^r, \bar t, L^r, \bar l) \in W$. In $G_1$, Spoiler picks a colored tree $H \in \{T^r,L^r\}$, index $o \in [s]$, and a vertex $v \in V(H)$. Erasing a corresponding pebble pair preserves a winning strategy: simulate the original strategy and ignore the erased pair until it is reused. Applying this observation to (iii), and noting that (i), (ii), and partial isomorphism are also preserved under erasure, we get $G_2:=\Game(T^r,\bar t[o \to *],L^r,\bar l[o \to *])\in W$. Put $m:=m_{\bar t[o \to *],\bar l[o \to *]}$, the restriction of $m_{\bar t,\bar l}$ to the indices $i$ of the branches $L_i^r$ of $L^r$ containing pebbles in $G_2$. The proof proceeds by case analysis.

    \begin{case}$v = r$\end{case}
    Then Duplicator responds in $G_1$ by picking $r$, obtaining $G_1^{1}:=\Game(T^r, \bar t[o \to r], L^r, \bar l[o \to r])$. Suppose that $\bar t[o \to r] \mapsto \bar l[o \to r]$ in $G_1^{1}$ is not an $s$-partial isomorphism. Then, one of the following occurs:
    \begin{itemize}
        \item The color of $r$ in $T^r$ is different from the color of $r$ in $L^r$, which is impossible by the definition of $L^r$.
        \item There is a placed pebble $j\ne o$ in $G_1$ with $t_j=r$ and $l_j\ne r$, or vice versa, which is impossible because $r$ has a unique color in both trees.
        \item There is a placed pebble $j\ne o$ in $G_1$ such that ($t_jr \in E(T^r)$ and $l_jr \not\in E(L^r)$) or ($t_jr \not\in E(T^r)$ and $l_jr \in E(L^r)$), hence the colors of $t_j$ and $l_j$ differ, which is impossible because $\bar t \mapsto \bar l$ in $G_1$ is an $s$-partial isomorphism.
    \end{itemize}
    The restrictions of the two tuples in $G_1^{1}$ to $T_i^r$ ($i\in[\ell]$) and $L_j^r$ ($j\in I$), respectively, coincide with those in $G_2$. Thus (i), (ii), and (iii) remain satisfied in $G_1^{1}$ with the same function $m$. Hence, $G_1^{1}\in W$, which completes the case of $v = r$.

    \begin{case}
        $v \neq r$ and $H=T^r$
    \end{case}
    Then, $v \in V(T^r_{j_T})$ for some index $j_T$.

    \begin{case2}
        In $G_2$, the branch $T^r_{j_T}$ contains a pebble
    \end{case2}
    Then, we fix an index $i\in[s]\setminus\{o\}$ with $t_i\in V(T^r_{j_T})$ in $G_1$, and we set $j_L$ to be the index of the subtree $L^r_{j_L}$ of $L^r$ containing $l_i$; this branch exists because the unique color of $r$ rules out $l_i=r$. By (i) and (ii) for $G_2$, $m(j_L)=j_T$, and by (iii) for $G_2$, Duplicator has a winning strategy in $G_1^{2.1}:=\Game(T^r_{j_T}, \bar t[o \to *][T^r_{j_T}], L^r_{j_L}, \bar l[o \to *][L^r_{j_L}])$. We play the move $(T^r_{j_T}, o, v)$ in $G_1^{2.1}$ as Spoiler. Following a winning strategy in $G_1^{2.1}$, Duplicator responds with a vertex $u$, yielding a game $G_2^{2.1}$. Then, we choose the same vertex $u$ for Duplicator's response in $G_1$, yielding $G_3^{2.1}:=\Game(T^r, \bar t[o \to v], L^r, \bar l[o \to u])$. Duplicator wins $G_2^{2.1}$, so $\bar t[o \to v][T^r_{j_T}] \mapsto \bar l[o \to u][L^r_{j_L}]$ is an $s$-partial isomorphism. Furthermore, no vertex of any other subtree is adjacent to $v$ or $u$, and the adjacency of $v$ and $u$ to $r$ depends only on the colors of $v$ and $u$. Hence, $\bar t[o \to v] \mapsto \bar l[o \to u]$ in $G_3^{2.1}$ is also an $s$-partial isomorphism. The function $m$ is unchanged from $G_2$, so (i) and (ii) remain satisfied in $G_3^{2.1}$. The restrictions of the two tuples in $G_3^{2.1}$ to $T_i^r$ ($i\in[\ell]\setminus\{j_T\}$) and $L_j^r$ ($j\in I\setminus\{j_L\}$), respectively, coincide with those in $G_2$. Together with the winning strategy in $G_2^{2.1}$, this gives (iii) in $G_3^{2.1}$. Thus $G_3^{2.1}\in W$.

    \begin{case2}
        In $G_2$, the branch $T^r_{j_T}$ contains no pebble
    \end{case2}
    Let $I_T := \{i\in[\ell] \mid L^r_{i} \cong L^r_{j_T}\}$ be the set of all indices of subtrees $L^r_{i}$ isomorphic to $L^r_{j_T}$, and let $I_L = I_T \cap I$ be its subset containing only indices of subtrees present in $L^r$. By construction, ${\min\{|I_T|, s\}=|I_L|}$. Let $I_T^* =\{ i \in I_T \mid \exists j\in[s]\setminus\{o\} : t_j \in V(T^r_{i}) \}$, and let $I_L^* =\{ i \in I_L \mid \exists j\in[s]\setminus\{o\} : l_j \in V(L^r_{i}) \}$. It follows from properties (i) and (ii) for $G_2$ that $|I_T^*|=|I_L^*|$. Since at most $s-1$ pebble pairs are placed in $G_2$, we get $|I_L^*| = |I_T^*| < s$. Since $j_T \not \in I_T^*$, we obtain $|I_L^*| = |I_T^*| < |I_T|$. Altogether, $|I_L^*| < |I_L|$. We choose an arbitrary index $j_L \in I_L \setminus I_L^*$. Since $L^r_{j_T}$ and $L^r_{j_L}$ are isomorphic and $L^r_{j_T}$ and $T^r_{j_T}$ are FO$^s$-equivalent, Duplicator wins $G_1^{2.2}:=\Game(T^r_{j_T}, *\ldots*, L^r_{j_L}, *\ldots*)$. We play the move $(T^r_{j_T}, o, v)$ in $G_1^{2.2}$ as Spoiler. Following a winning strategy in $G_1^{2.2}$, Duplicator responds with a vertex $u$, yielding a game $G_2^{2.2}$. Then, we choose the same vertex $u$ for Duplicator's response in $G_1$, yielding $G_3^{2.2}:=\Game(T^r, \bar t[o \to v], L^r, \bar l[o \to u])$. Again, no vertex of any other subtree is adjacent to $v$ or $u$, and the adjacency of $v$ and $u$ to $r$ depends only on the colors of $v$ and $u$. Together with the partial isomorphism in $G_2^{2.2}$, this shows that $\bar t[o \to v] \mapsto \bar l[o \to u]$ in $G_3^{2.2}$ is an $s$-partial isomorphism. Both chosen branches $T^r_{j_T}$ and $L^r_{j_L}$ do not contain pebbles in $G_2$, so we can extend $m$ by $m(j_L):=j_T$ while preserving injectivity. The trees $L^r_{j_T}$ and $L^r_{j_L}$ are isomorphic, giving (i) and (ii) in $G_3^{2.2}$. The restrictions of the two tuples in $G_3^{2.2}$ to $T_i^r$ ($i\in[\ell]\setminus\{j_T\}$) and $L_j^r$ ($j\in I\setminus\{j_L\}$), respectively, coincide with those in $G_2$. Together with the winning strategy in $G_2^{2.2}$, this gives (iii) in $G_3^{2.2}$. Thus $G_3^{2.2}\in W$.

    \begin{case}[store=case:reduce-graph:deferred]\label{case:reduce-graph:deferred}
        $v \neq r$ and $H=L^r$
    \end{case}
%
%

            Then, $v \in V(L^r_{j_L})$ for some index $j_L$.

            \begin{case2}
                In $G_2$, the branch $L^r_{j_L}$ contains a pebble
            \end{case2}

            Then, we fix an index $i\in[s]\setminus\{o\}$ with $l_i\in V(L^r_{j_L})$ in $G_1$, and we set $j_T$ to be the index of the subtree $T^r_{j_T}$ of $T^r$ containing $t_i$; this branch exists because the unique color of $r$ rules out $t_i=r$. By (i) and (ii) for $G_2$, $m(j_L)=j_T$, and by (iii) for $G_2$, Duplicator has a winning strategy in $G_1^{3.1}:=\Game(T^r_{j_T}, \bar t[o \to *][T^r_{j_T}], L^r_{j_L}, \bar l[o \to *][L^r_{j_L}])$. We play the move $(L^r_{j_L}, o, v)$ in $G_1^{3.1}$ as Spoiler. Following a winning strategy in $G_1^{3.1}$, Duplicator responds with a vertex $u$, yielding a game $G_2^{3.1}$. Then, we choose the same vertex $u$ for Duplicator's response in $G_1$, yielding $G_3^{3.1}:=\Game(T^r, \bar t[o \to u], L^r, \bar l[o \to v])$. Duplicator wins $G_2^{3.1}$, so $\bar t[o \to u][T^r_{j_T}] \mapsto \bar l[o \to v][L^r_{j_L}]$ is an $s$-partial isomorphism. Furthermore, no vertex of any other subtree is adjacent to $v$ or $u$, and the adjacency of $v$ and $u$ to $r$ depends only on the colors of $v$ and $u$. Hence, $\bar t[o \to u] \mapsto \bar l[o \to v]$ in $G_3^{3.1}$ is also an $s$-partial isomorphism. The function $m$ is unchanged from $G_2$, so (i) and (ii) remain satisfied in $G_3^{3.1}$. The restrictions of the two tuples in $G_3^{3.1}$ to $T_i^r$ ($i\in[\ell]\setminus\{j_T\}$) and $L_j^r$ ($j\in I\setminus\{j_L\}$), respectively, coincide with those in $G_2$. Together with the winning strategy in $G_2^{3.1}$, this gives (iii) in $G_3^{3.1}$. Thus $G_3^{3.1}\in W$.

            \begin{case2}
                In $G_2$, the branch $L^r_{j_L}$ contains no pebble
            \end{case2}

            Let $I_T := \{i\in[\ell] \mid L^r_{i} \cong L^r_{j_L}\}$ be the set of all indices of subtrees $L^r_{i}$ isomorphic to $L^r_{j_L}$, and let $I_L = I_T \cap I$ be its subset containing only indices of subtrees present in $L^r$.
            Let $I_T^* =\{ i \in I_T \mid \exists j\in[s]\setminus\{o\} : t_j \in V(T^r_{i}) \}$, and let $I_L^* =\{ i \in I_L \mid \exists j\in[s]\setminus\{o\} : l_j \in V(L^r_{i}) \}$. It follows from properties (i) and (ii) for $G_2$ that $|I_T^*|=|I_L^*|$.
            Since $j_L \not \in I_L^*$, we get  $|I_T^*| = |I_L^*| < |I_L|$.
            Since $I_L \subseteq I_T$, we get $|I_L| \le |I_T|$. Altogether, $|I_T^*| < |I_T|$.

            We choose an arbitrary index $j_T \in I_T \setminus I_T^*$. Since $L^r_{j_L}$ and $L^r_{j_T}$ are isomorphic and $L^r_{j_T}$ and $T^r_{j_T}$ are FO$^s$-equivalent, Duplicator wins $G_1^{3.2}:=\Game(T^r_{j_T}, *\ldots*, L^r_{j_L}, *\ldots*)$. We play the move $(L^r_{j_L}, o, v)$ in $G_1^{3.2}$ as Spoiler. Following a winning strategy in $G_1^{3.2}$, Duplicator responds with a vertex $u$, yielding a game $G_2^{3.2}$. Then, we choose the same vertex $u$ for Duplicator's response in $G_1$, yielding $G_3^{3.2}:=\Game(T^r, \bar t[o \to u], L^r, \bar l[o \to v])$. Again, no vertex of any other subtree is adjacent to $v$ or $u$, and the adjacency of $v$ and $u$ to $r$ depends only on the colors of $v$ and $u$. Together with the partial isomorphism in $G_2^{3.2}$, this shows that $\bar t[o \to u] \mapsto \bar l[o \to v]$ in $G_3^{3.2}$ is an $s$-partial isomorphism. Both chosen branches $T^r_{j_T}$ and $L^r_{j_L}$ do not contain pebbles in $G_2$, so we can extend $m$ by $m(j_L):=j_T$ while preserving injectivity. The trees $L^r_{j_T}$ and $L^r_{j_L}$ are isomorphic, giving (i) and (ii) in $G_3^{3.2}$. The restrictions of the two tuples in $G_3^{3.2}$ to $T_i^r$ ($i\in[\ell]\setminus\{j_T\}$) and $L_j^r$ ($j\in I\setminus\{j_L\}$), respectively, coincide with those in $G_2$. Together with the winning strategy in $G_2^{3.2}$, this gives (iii) in $G_3^{3.2}$. Thus $G_3^{3.2}\in W$.
\end{proof}

Using \Cref{thm:reduce-graph}, we deduce a model-checking algorithm.

\begin{retheorem}[store=thm:tree-algo,label=thm:tree-algo]
    There exist a computable function $f: \NN \times \NN \times \NN \to \NN$, and an algorithm $\AA$ such that, given an FO sentence $\phi$ with $s$ variables and a $c$-colored tree $T$ of depth $k$, $\AA$ decides whether $T \models \phi$ in time $\ca O(f(s, k, c) \cdot (|T|+|\phi|))$.
\end{retheorem}
\begin{proof}
    Let $s \in \NN$, let $T$ be a $c$-colored $n$-vertex tree of depth $k$, and let $\phi$ be an FO$^s$ sentence. We use the algorithm of \Cref{thm:reduce-graph} to obtain a $c$-colored tree $T_0 \equiv^{\text{FO}^s} T$  on at most $g(s, k, c)$ vertices in time $\ca O(h(s, k, c) \cdot n)$. In particular, $T \models \phi$ if and only if $T_0 \models \phi$. We denote $V_0 := V(T_0)$ and $n_0 := |V_0| \le g(s, k, c)$.

    Finally, we apply the naive evaluation algorithm to decide whether $T_0 \models \phi$. The algorithm computes, for each subformula $\psi$ of $\phi$ with $q$ free variables, the set $R_\psi$ of $q$-tuples $\bar v \in V_0^q$ satisfying $T_0 \models \psi(\bar v)$. This takes $\ca O(n_0^{\min\{q+1,s\}}) \le \ca O(g(s, k, c)^s)$ time. Altogether, the running time of our algorithm is $\ca O(h(s, k, c) \cdot n + |\phi| \cdot g(s, k, c)^s)$.
\end{proof}

\Cref{thm:tree-algo} naturally generalizes from shallow colored trees to (colored) graphs of bounded tree-depth and SC-depth, and thus also to classes of bounded shrub-depth. Elimination forests \cite{sparsity-book} and tree-models \cite{DBLP:conf/stacs/GajarskyK20} can be computed in FPT time.
A graph can be interpreted from a coloring of its elimination forest, where the root has a unique color, and the colors of all the other vertices encode their adjacencies to their ancestors. Similarly, a graph can be interpreted from its tree-model \cite{DBLP:journals/lmcs/GanianHNOM19} by encoding in its colors the table that determines the adjacency from the colors and distances of its leaves.
The Backwards Translation Lemma tells us how to use the interpretation to modify a sentence so that we can evaluate it on the coloring of the underlying elimination forest or tree-model instead of the input graph.

\begin{recorollary}[store=col:algo-tree-depth,label=col:algo-tree-depth]
    There exist a computable function $f: \NN \times \NN \to \NN$, and an algorithm $\AA$ such that, given an FO sentence $\phi$ with $s$ variables and a graph $G$ of tree-depth $k$, $\AA$ decides whether $G \models \phi$ in time $\ca O(f(s, k) \cdot (|G|+|\phi|))$.
\end{recorollary}

\begin{theorem}[Gajarsk\'y and Kreutzer {\cite[Theorem~5.5]{DBLP:conf/stacs/GajarskyK20}}]\label{thm:compute-tree-model}
    There are an absolute constant $c\ge1$, a computable function $g: \NN\times\NN\to\NN$, and an algorithm that, given a graph $G$ and positive integers $d,m$, computes a tree-model of $G$ of depth $d$ with at most $m$ colors, or correctly reports that none exists, in time $\ca O(g(d,m)|G|^c)$.
\end{theorem}

\begin{corollary}\label{col:algo-sc-depth}
    There exist a computable function $f: \NN \times \NN \to \NN$, and an algorithm $\AA$ such that, given an FO sentence $\phi$ with $s$ variables and a graph $G$ of SC-depth $k$, $\AA$ decides whether $G \models \phi$ in time $\ca O(f(s, k) \cdot (|G|^c+|\phi|))$, where $c$ is the constant from \Cref{thm:compute-tree-model}.
\end{corollary}

\begin{corollary}\label{col:algo-shrub-depth}
    For every graph class $\ca C$ of bounded shrub-depth, there exist a computable function $f_{\ca C}: \NN \to \NN$ and an algorithm $\AA$ such that, for every integer $s \in \NN$, given an FO sentence $\phi$ with $s$ variables and a graph $G \in \ca C$, $\AA$ decides whether $G \models \phi$ in time $\ca O(f_{\ca C}(s) \cdot (|G|^c+|\phi|))$, where $c$ is the constant from \Cref{thm:compute-tree-model}.
\end{corollary}

\section{Hardness}\label{sec:hardness}

We consider both the monotone and hereditary settings, and we investigate which graph classes admit FO model checking in \FPT-time when parameterized by the number of variables rather than quantifier rank. In the monotone setting, the following theorem extends \Cref{col:algo-tree-depth} to a full characterization of these classes.

\getkeytheorem{thm:unbounded-tree-depth-hardness}

In the hereditary setting, we provide strong evidence suggesting that classes of bounded shrub-depth are the ``limit'' of this parameterization of FO model checking. The key difference from the monotone setting is that, for a monotone class $\ca C$ of unbounded tree-depth, it is easy to construct a graph $P \in \ca C$ of arbitrarily large tree-depth. Indeed, it suffices to take a long path $P$. M{\"{a}}hlmann \cite{DBLP:conf/icalp/Mahlmann25} proved that every hereditary class $\ca C$ has unbounded shrub-depth if and only if $\ca C$ contains flipped half-graphs of order $t$ for all $t \ge 1$, or $\ca C$ contains layerwise flipped $tP_t$ for all $t \ge 1$.
We do not know how to algorithmically construct a layerwise flip of $tP_t$ belonging to $\ca C$. Nevertheless, we obtain the following conditional result.

\getkeytheorem{thm:unbounded-shrub-depth-hardness}

In both settings, we reduce to paths.

\getkeytheorem{thm:paths-hard}

Before proving \Cref{thm:paths-hard}, we prove two simple lemmas.

\begin{lemma}\label{lem:distance-formula}
    For every $k \in \mathbb{N}$, there is a formula $\xi_k(x_1, x_2)$ of length $\ca O(k)$ using only variables $x_1$, $x_2$, $x_3$, and $x_4$ such that, for every path $P$ and vertices $u,v \in V(P)$, it holds that $P \models \xi_k(u,v)$ if and only if the distance in $P$ between $u$ and $v$ is exactly $k$.

    Furthermore, given $k$, the formula $\xi_k$ can be computed in time $\ca O(k)$.
\end{lemma}
\begin{proof}
    For every $k \ge 1$ and every permutation $\sigma:[4] \to [4]$, we construct a formula $\chi^\sigma_k(x_{\sigma(1)}, x_{\sigma(2)}, x_{\sigma(3)})$ using only the variables $x_1$, $x_2$, $x_3$, and $x_4$ such that, for every graph $G$ and vertices $a, b, c \in V(G)$, it holds that $G \models \chi^{\sigma}_k(a, b, c)$ if and only if there is a walk $(w_0=a, w_1=b, w_2, \ldots, w_{k-1}, w_k=c)$ in $G$ such that, for every $0 < i < k$, we have $w_{i-1} \neq w_{i+1}$.

    Let $\omega: [4] \to [4]$ be the permutation given by $\omega(1)=2$, $\omega(2)=4$, $\omega(3)=3$, and $\omega(4)=1$. We consider the following inductively defined formulas:
    \begin{align*}
    \chi^\sigma_1 &\equiv E(x_{\sigma(1)}, x_{\sigma(2)}) \land x_{\sigma(2)} = x_{\sigma(3)} \\
    \chi^\sigma_{k+1} &\equiv E(x_{\sigma(1)}, x_{\sigma(2)}) \land \exists x_{\sigma(4)} : x_{\sigma(1)} \neq x_{\sigma(4)} \land E(x_{\sigma(2)}, x_{\sigma(4)}) \land \chi^{\sigma\circ\omega}_{k}
    \end{align*}

    We observe that the formula $\chi^\sigma_k$ expresses precisely these conditions. If $G$ is a path, then any walk satisfying these conditions is also a path, so we complete the proof by setting $\xi_0 \equiv x_1=x_2$ and $\xi_k \equiv \exists x_3 : \chi_k^{\rho}$ where $\rho(1)=1$, $\rho(2)=3$, $\rho(3)=2$, and $\rho(4)=4$.
\end{proof}

\begin{lemma}\label{lem:hardcode-graph}
    Let $G$ be an $n$-vertex graph. Then, there is a formula $\epsilon_G(x_1, x_2, x_3)$ of length $\ca O(n^3)$ using only the variables $x_1$, $x_2$, $x_3$, and $x_4$ such that, for every $n$-vertex path $P$ and its endpoint $p \in V(P)$, there is a bijection $f: V(G) \to V(P)$ such that, for every pair of vertices $u,v \in V(G)$, it holds that $P \models \epsilon_G(p, f(u), f(v))$ if and only if $uv \in E(G)$.

    Furthermore, given $G$, the formula $\epsilon_G$ can be computed in time $\ca O(n^3)$.
\end{lemma}
\begin{proof}
    Let $v_1, \ldots, v_n$ be an arbitrary ordering of the vertices of $G$.
    For each $k \in \mathbb{N}$, let $\xi_k(x_1, x_2)$ be the formula from \Cref{lem:distance-formula}. Consider the following formula.
\[\epsilon_G \equiv \bigvee\limits_{\substack{1 \le i,j \le n \\ v_iv_j \in E(G)}} \left(\xi_{i-1}(x_1, x_2) \land \xi_{j-1}(x_1, x_3)\right)\]

    Informally, $\epsilon_G$ hardcodes the adjacency relation of $G$, and it encodes each vertex $v_i$ of $G$ as the unique vertex at distance $i-1$ from $x_1$ in an $n$-vertex path with endpoint $x_1$.

    Let $P$ be an arbitrary $n$-vertex path with endpoint $p$, and let $p_1=p, p_2, \ldots, p_n$ be the natural ordering of the vertices of $P$ along the path. Observe that the bijection $f(v_i)=p_i$ has the desired property, that is, for every pair of vertices $u,v \in V(G)$, it holds that $P \models \epsilon_G(p, f(u), f(v))$ if and only if $uv \in E(G)$.
\end{proof}

We now prove \Cref{thm:paths-hard} by reduction from the following classical result.

\begin{theorem}[Downey, Fellows, and Taylor {\cite[Proposition~2]{DBLP:conf/dmtcs/DowneyFT96}}]\label{thm:all-graph-hard}
    The FO model checking problem parameterized by quantifier rank and restricted to the class of all simple graphs is \AW{*}-hard.
\end{theorem}

\begin{retheorem}[store=thm:paths-hard,label=thm:paths-hard]
    The FO model checking problem parameterized by the number of variables and restricted to the class of all paths is \AW{*}-hard.
\end{retheorem}
%
\begin{proof}
    Let $G$ be an $n$-vertex graph and let $\phi$ be an FO sentence with quantifier rank $q$. Assume that $n \ge 3$.
    Without loss of generality, we assume that $\phi$ only uses the variables $x_2, x_3, \ldots, x_{q+1}$; otherwise, we could rename the variables while preserving logical equivalence.

    Let $\epsilon_G$ be the formula from \Cref{lem:hardcode-graph} which encodes the adjacency matrix of $G$. Given $i \neq j \in \{2,3, \ldots, q+1\}$, we denote by $\epsilon_G^{i,j}$ the formula obtained from $\epsilon_G$ by simultaneously replacing $x_2$ by $x_i$, $x_3$ by $x_j$, and $x_4$ by $x_s$ where $s = \min\left(\{2,3,4\}\setminus\{i, j\}\right)$. 
    Let $\psi'_G(x_1)$ be a formula obtained from $\phi$ by replacing each atom $E(x_i, x_j)$ by the formula $\epsilon_G^{i,j}$ if $i \neq j$, and by $x_i \neq x_i$ if $i=j$. Set $\psi_G:=\exists x_1 : \psi'_G \land \exists x_2 \forall x_3 : E(x_1, x_3) \rightarrow x_2=x_3$.

    Let $P$ be an $n$-vertex path. Let $p$ be an endpoint of $P$, and let $f: V(G) \to V(P)$ be a bijection (from \Cref{lem:hardcode-graph}) such that, for every pair of vertices $u,v \in V(G)$, it holds that
    \begin{itemize}
        \item[(a)] $uv \in E(G)$ if and only if $P \models \epsilon_G(p, f(u), f(v))$.
    \end{itemize}

    \begin{reclaim}[store=claim:parainterpr,label=claim:parainterpr]
        $G \models \phi$ if and only if $P \models \psi'_G(p)$
    \end{reclaim}
    \begin{subproof}
            We show that $G \models \phi$ if and only if $P \models \psi'_G(p)$ by structural induction on $\phi$.

            Let $\xi$ be a subformula of $\phi$, and let $\xi_G$ be the corresponding subformula of $\psi'_G$. Let $c: \textsf{free}(\xi) \to V(G)$ be any variable assignment, where $\textsf{free}(\xi)$ denotes the set of free variables of $\xi$. Denote by $c^p$ the variable assignment defined by $c^p(x_1)=p$ and $c^p(x_i)=f(c(x_i))$ for $x_i \in \textsf{free}(\xi) \subseteq \{x_2, \ldots, x_{q+1}\}$.
            We show that $G, c \models \xi$ if and only if $P, c^p \models \xi_G$.
            \begin{itemize}
                \item If $\xi$ is of the form $E(x_i, x_i)$, then $\xi_G$ is $x_i \neq x_i$, and the equivalence follows from the simplicity of $G$.
                \item If $\xi$ is of the form $E(x_i, x_j)$ where $i \neq j$, then $\xi_G$ is $\epsilon_G^{i,j}(x_1, x_i, x_j)$, so the equivalence follows from (a).
                \item If $\xi$ is of the form $x_i = x_j$, then $\xi_G = \xi$, and the equivalence follows from the bijectivity of $f$.
                \item If $\xi$ is of the form $\exists x_i : \chi$, then $\xi_G$ is of the form $\exists x_i \chi_G$, and the equivalence follows by applying the induction hypothesis to $\chi$ and $\chi_G$ for each variable assignment obtained from $c$ by varying $x_i$ over all vertices of $G$.
                \item If $\xi$ is a Boolean combination of subformulas, then the equivalence follows immediately from the induction hypothesis.
            \end{itemize}
            Hence, $G \models \phi$ if and only if $P \models \psi'_G(p)$.
    \end{subproof}

    If $G \models \phi$, then $P \models \psi_G$ as well, since we can assign $p$ to $x_1$, satisfying both $\exists x_2 \forall x_3 : E(x_1, x_3) \rightarrow x_2=x_3$ and $\psi'_G$.

    If $P \models \psi_G$, then one of the endpoints of $P$ must be assigned to $x_1$. Since there is an isomorphism of $P$ swapping the endpoints of $P$, we assume that $p$ is the endpoint assigned to $x_1$. We obtain $G \models \phi$, because $P \models \psi'_G(p)$.

    Altogether, $G \models \phi$ if and only if $P \models \psi_G$. Hence, the mapping $(G, \phi) \mapsto (P, \psi_G)$ is the desired reduction.
    
\end{proof}

We are now ready to prove Theorems~\ref{thm:unbounded-tree-depth-hardness}~and~\ref{thm:unbounded-shrub-depth-hardness}.

\begin{retheorem}[store=thm:unbounded-tree-depth-hardness,label=thm:unbounded-tree-depth-hardness]
    Let $\ca C$ be a monotone class of unbounded tree-depth. Then, the \textsc{FO model-checking} problem is \AW{*}-hard when restricted to $\ca C$ and parameterized by the number of variables in the formula.
\end{retheorem}
\begin{proof}
    The class $\ca C$ has unbounded tree-depth, so $\ca C$ contains all paths as subgraphs of graphs from $\ca C$. It follows from monotonicity that the class of all paths is a subclass of $\ca C$. Hence, \AW{*}-hardness follows from \Cref{thm:paths-hard}.
\end{proof}

In the hereditary setting, we use the following result of M{\"{a}}hlmann \cite{DBLP:conf/icalp/Mahlmann25}.

\begin{theorem}[M{\"{a}}hlmann {\cite[Theorem~1]{DBLP:conf/icalp/Mahlmann25}}]\label{thm:unbounded-shrub-depth-interprets-paths}
    For every hereditary graph class $\ca C$, the following are equivalent.
    \begin{itemize}
        \item $\ca C$ has bounded shrub-depth.
        \item There is $t \ge 1$ such that $\ca C$ excludes all flipped half-graphs of order $t$ and all layerwise flipped $tP_t$.
        \item There is $t \ge 1$ such that $\ca C$ excludes all flipped half-graphs of order $t$ and all layerwise flipped $3P_t$.
        \item $\ca C$ does not (1-dimensionally) FO-interpret the class of all paths.
    \end{itemize}
\end{theorem}

M{\"{a}}hlmann \cite{DBLP:conf/icalp/Mahlmann25} proves the last item by giving an explicit interpretation independent of the class $\ca C$.

\begin{lemma}[M{\"{a}}hlmann {\cite[Lemmas~31 and~34]{DBLP:conf/icalp/Mahlmann25}}]\label{lem:constructive-path-interpretation}
    There exists an interpretation $\iota$ with the following properties.
    \begin{itemize}
        \item For every $t\ge5$ and every flipped half-graph $H$ of order $t$, there is an induced subgraph $W^H_t$ of $H$ such that $\iota(W^H_t)$ is a path on $t-4$ vertices.
        \item For every $t\ge4$, there is an induced subgraph $Q_t$ of $5P_t$ such that $\iota(F\oplus Q_t)$ is a path on $t$ vertices for every layerwise flip $F$.
    \end{itemize}
\end{lemma}

The construction of $W^H_t$ and $Q_t$ selects vertices of $H_t$ and $5P_t$ according to a simple pattern. Thus, given $H_t$ or $5P_t$, an algorithm can output a set of vertices that induces $W^H_t$ or $Q_t$.
We refer to the proofs in \cite[Lemma~31]{DBLP:conf/icalp/Mahlmann25} and \cite[Lemma~4.9]{DBLP:journals/corr/abs-2501-13903} for details.
\begin{retheorem}[store=thm:unbounded-shrub-depth-hardness,label=thm:unbounded-shrub-depth-hardness]
    Let $\ca C$ be a hereditary class of unbounded shrub-depth. Assume that
    one of the following holds:
    \begin{enumerate}[(a)]
        \item\label{item:has-flipped-half-graphs} For every $t \ge 1$, $\ca C$ contains a flipped half-graph of order $t$; or
        \item\label{item:algo-flip-of-tPt} there is a polynomial-time algorithm $\AA_{\ca C}$ that, given a unary representation of an integer $t \ge 1$, produces a layerwise flip $F$ of $tP_t$ such that $F \oplus tP_t \in \ca C$.
    \end{enumerate}
    Then, the \textsc{FO model-checking} problem is \AW{*}-hard when restricted to $\ca C$ and parameterized by the number of variables in the formula.
\end{retheorem}
\begin{proof}
    We start by showing that there is an algorithm that, given a unary representation of an integer $t\ge5$, produces a graph $G_t$ such that $\iota(G_t)$ is a path on $t$ vertices, where $\iota$ is the interpretation from \Cref{lem:constructive-path-interpretation}.

    If Item~(\ref{item:has-flipped-half-graphs}) holds, then $\ca C$ contains one of the eight possible flips of the half-graph infinitely many times, so there exists an algorithm that, given an integer $t$, outputs the flip $F$ of a half-graph $H_{t+4}$ of order $t+4$ such that $F \oplus H_{t+4} \in \ca C$.
    We construct the induced subgraph $W^{F \oplus H_{t+4}}_{t+4}$. This gives a graph $G_t:=W^{F\oplus H_{t+4}}_{t+4} \in \ca C$ such that the $\iota$-interpretation of $G_t$ is a path on $t$ vertices.

    Otherwise, we use the algorithm from Item~(\ref{item:algo-flip-of-tPt}) to obtain the flip $F$ of $tP_t$ such that $F \oplus tP_t \in \ca C$. We construct the induced subgraph $Q_t$. We again obtain a graph $G_t:=F \oplus Q_t \in \ca C$ such that the $\iota$-interpretation of $G_t$ is a path on $t$ vertices.

    Let $\phi$ be an FO$^s$ sentence, and let $P$ be a path on $t$ vertices. If $t < 5$, we can use the naive evaluation algorithm to decide  whether $P \models \phi$. Hence, we assume that $t \ge 5$. Since $P$ is isomorphic to $\iota(G_t)$, it follows from the Backwards Translation Lemma that there is a formula $\phi[\iota]$ such that $P \models \phi$ if and only if $G_t \models \phi[\iota]$. Furthermore, the number of variables of $\phi[\iota]$ is at most $s+c$, where $c$ is a constant depending only on $\iota$. Hence, mapping $(\phi, P) \mapsto (\phi[\iota], G_t)$ is the desired reduction. Thus, \AW{*}-hardness follows from \Cref{thm:paths-hard}.
\end{proof}

We also prove that the algorithmic results from \Cref{sec:algo} cannot be extended to MSO parameterized by the number of variable names, unless \FPT$=$\AW{*}.

\begin{retheorem}[store=thm:MSO-aw-hard,label=thm:MSO-aw-hard]
    The \textsc{MSO model-checking} problem parameterized by the number of variables is \AW{*}-hard when restricted to disjoint unions of stars.
\end{retheorem}
\begin{proof}
    For each $i \in \NN$, we describe a formula $\phi_i(X_1)$ expressing ``$|X_1|=i$'' using only set variables $X_1$ and $X_2$ and vertex variables $x_1$ and $x_2$.
    \begin{align*}
        \phi_0 &\equiv \forall x_1 : x_1 \not\in X_1 \\
        \phi_{i+1} &\equiv \exists X_2 \exists x_2 : (x_2 \not\in X_2 \land x_2 \in X_1) \land\\
        & \hspace{7mm} \bigl[\forall x_1: (x_1 \in X_2 \rightarrow x_1 \in X_1) \land (x_1 \in X_1 \rightarrow (x_1 \in X_2 \lor x_1 = x_2)) \bigr] \land \\
        & \hspace{7mm} \bigl[ \exists X_1 \forall x_1 : (x_1 \in X_1 \leftrightarrow x_1 \in X_2) \land \phi_i \bigr]
    \end{align*}
    The formula $\phi_0$ expresses that $X_1$ is empty set, and the formula $\phi_{i+1}$ expresses that removing one element from the set $X_1$ gives a set satisfying $\phi_i$.

    For both $j=1,2$, we define a formula $\xi_i^j(x_j)$ expressing ``$x_j$ has $i$ neighbors'':
    \[\xi_i^j \equiv \exists X_1 \forall x_{3-j} : (x_{3-j} \in X_1 \leftrightarrow E(x_j, x_{3-j})) \land \phi_i\]

    Let $G$ be a graph on vertex set $V(G)=\{v_1, v_2, \ldots, v_n\}$. We use the formula $\xi_i^j$ to define a formula $\epsilon_G(x_1, x_2)$ hardcoding the adjacency matrix of $G$.
    \[\epsilon_G \equiv \bigvee\limits_{\substack{1 \le i,j \le n \\ v_iv_j \in E(G)}} \left(\xi_{i+2}^1 \land \xi_{j+2}^2\right)\]

    Let $J_n$ be the disjoint union of stars $K_{1,3}, K_{1,4}, \ldots, K_{1,n+2}$. Then, $J_n \models \epsilon_G(x_1, x_2)$ if and only if $x_1$ is the center of $K_{1,i+2}$ and $x_2$ is the center of $K_{1,j+2}$, and $v_iv_j \in E(G)$.
    We denote by $\epsilon_G^{i,j}$ the formula obtained from $\epsilon_G$ by simultaneously renaming $x_1$ to $x_i$ and $x_2$ to $x_j$.
    
    Finally, we describe the reduction from the FO model checking problem of \Cref{thm:all-graph-hard}.
    Let $\psi$ be an FO sentence. Let $\psi_G$ be a sentence obtained from $\psi$ by restricting quantification to vertices of degree at least 2, and replacing $E(x_i, x_j)$ by $\epsilon_G^{i,j}$.
    Then, we get that $G \models \psi$ if and only if $J_n \models \psi_G$. Furthermore, $|\psi_G|\in\ca O(n^3|\psi|)$ and $|J_n|\in\ca O(n^2)$.
    Hence, $(G, \psi) \mapsto (J_n, \psi_G)$ is the desired reduction.
\end{proof}

\section{Applications}\label{sec:app}

We discuss two applications of our algorithmic results. We characterize classes of bounded shrub-depth in terms of finiteness of the number of distinct FO$^s$ types. We also derive an \FPT-time algorithm for deciding the existence of a homomorphism from a graph $H$ to $G$, parameterized by tree-width of $H$ and SC-depth of $G$.

The characterization of classes of bounded shrub-depth is a corollary of \Cref{thm:reduce-graph}:
\begin{recorollary}[store=cor:shrub-depth-unbounded-num-types,label=cor:shrub-depth-unbounded-num-types]
    There is a constant $s_0 \in \NN$ such that the following holds. Fix $k \in\mathbb{N}$, let $\ca C$ be a hereditary class of $k$-colored graphs, and let $s \in \NN, s \ge s_0$. Denote by $\tp^s(\ca C):=\{\tp^s(G) \mid G \in \ca C\}$ the set of all FO$^s$ types of graphs from $\ca C$. Then, $\ca C$ has bounded shrub-depth if and only if there exists $b \in \NN$ such that $|\tp^s(\ca C)| \le b$.
\end{recorollary}
    \begin{proof}{}\,{}
        $(\implies)$ Let $\ca C$ be a class of $k$-colored graphs with bounded shrub-depth. Then, there exist an interpretation $\iota_{\scriptscriptstyle\ca C}$ and constants $d$ and $c$ such that $\ca C \subseteq \iota_{\scriptscriptstyle\ca C}(\ca T_{d, c+k})$, where $\ca T_{d,\ell}$ denotes the class of all $\ell$-colored trees of depth $d$.
        Observe that it suffices to show that the number of FO$^{s+s'}$ types of trees from $\ca T_{d,c+k}$ is bounded for some constant $s'$ depending only on the interpretation $\iota_{\scriptscriptstyle\ca C}$. This follows immediately from \Cref{thm:reduce-graph}, because each tree from $\ca T_{d,c+k}$ is FO$^{s+s'}$ equivalent to one of the finitely many $(c+k)$-colored trees on a bounded number of vertices.

        $(\impliedby)$ Let $\iota$ be the interpretation from \Cref{lem:constructive-path-interpretation}. Note that $\iota$ is a fixed interpretation independent of $\ca C$, and it does not refer to the colors. Let $\ca C$ be a hereditary graph class with unbounded shrub-depth. Then, $\iota(\ca C)$ contains the class of all paths. We observe that Spoiler wins the FO$^3$ pebble game between any pair of paths of distinct lengths, so the class of all paths has an unbounded number of FO$^3$ types. Hence, $\ca C$ also has an unbounded number of FO$^{s_0}$ types for some $s_0$ depending only on the interpretation $\iota$.
    \end{proof}

\Cref{cor:shrub-depth-unbounded-num-types} gives us another tool for proving that a graph class $\ca C$ has unbounded shrub-depth, because it suffices to show that there is an infinite subclass $\ca D \subseteq \ca C$ such that Spoiler wins the FO$^s$ pebble game between any pair of distinct structures from $\ca D$ for some constant $s$.

M{\"{a}}hlmann \cite{DBLP:conf/icalp/Mahlmann25} asked if the following are equivalent for every hereditary graph class $\ca C$:
\begin{itemize}
    \item $\ca C$ has bounded shrub-depth.
    \item There is $t \ge 1$ such that $\ca C$ excludes all flipped half-graphs of order $t$ and all layerwise flipped $2P_t$.
\end{itemize}

One direction is immediate -- a layerwise flipped $tP_t$ also contains a layerwise flipped $2P_t$. The other direction follows from \Cref{cor:shrub-depth-unbounded-num-types}. M{\"{a}}hlmann told us in private communication that another solution of this question has already been submitted. To our knowledge, that solution is not yet publicly available. Nevertheless, we include the following proof for its simplicity and to illustrate the use of \Cref{cor:shrub-depth-unbounded-num-types}.

\begin{lemma}\label{lem:2Pt-unbnd-shrub-depth}
    Let $\ca C$ be a class containing a layerwise flip of $2P_t$ for all $t$. Then, $\ca C$ has unbounded shrub-depth.
\end{lemma}
\begin{proof}
    For each $t \ge 1$, we fix a graph $H_t \in \ca C$ such that $H_t$ is a layerwise flip of $2P_t$, and we fix its layering along the paths $\ca L^t=(L_1^t, L_2^t, \ldots, L^t_t)$.
    We denote by $H_t^*$ the colored graph obtained by coloring two vertices in $L_1^t$ with color 1 and all other vertices with color 2. We observe that $\ca C^* := \{G \mid G \text{ is an induced subgraph of } H_t^* \text{ for some } t \ge 1 \}$ is hereditary and transducible from $\ca C$.
    Hence, it suffices to show that $\ca C^*$ has unbounded shrub-depth. We do so by proving that, for all $t > \ell > 2$, Spoiler wins the FO$^s$ pebble game between $H^*_t$ and $H^*_\ell$ for $s=6$.
    
    Observe that the set $S(G; x, y) := \{z \in V(G) \mid z \not\in \{x, y\} \land (xz \in E(G) \mathbin{\mathrm{xor}} yz \in E(G))\}$ is FO-definable, and for $x \neq y \in L_i^t$, we have $S(H_t^*; x, y) = L^t_{i-1} \cup L^t_{i+1}$, where $L_0^t=L_{t+1}^t=\emptyset$. This gives Spoiler a winning strategy in the game on $H^*_t$ and $H^*_\ell$. During the game, we will maintain the following invariant: if Spoiler places a pebble in $L^t_i$, then Duplicator must place the corresponding pebble in $L^\ell_i$.

    First, Spoiler places two pebbles in $L_1^t$, and Duplicator must respond by placing two pebbles in $L_1^\ell$ because of the color of these vertices. Then, for each $i \ge 2$, Spoiler keeps the pebbles on $L^t_{i-1} \cup L^t_{i-2}$, and moves a pair of pebbles from $L^t_{i-3}$ to $L^t_{i}$ (or places two new pebbles there if $i \le 3$).  Duplicator then moves or places a pebble on a vertex $v$. If $v \in L_{i-1}^\ell \cup L_{i-2}^\ell$, then it follows from the invariant that Duplicator loses immediately. If $v \not\in S(H_\ell^*; L_{i-1}^\ell)$, then Duplicator also loses immediately because of different adjacencies to the pebbles placed on the $(i-1)$-th layer. Hence, Duplicator places the pebble on a vertex of $S(H_\ell^*; L_{i-1}^\ell) \setminus (L_{i-1}^\ell \cup L_{i-2}^\ell) = L_i^\ell$.

    Since $t > \ell$, the above analysis also gives us that, once Spoiler places a pebble on the $(\ell+1)$-th layer, Duplicator has no response that avoids losing.
\end{proof}

\medskip

The graph homomorphism problem asks, given two graphs $G$ and $H$, whether $H$ is homomorphic to $G$.
An algorithm of D{\'{\i}}az et al. \cite{DBLP:journals/tcs/DiazST02} decides this problem in \XP-time parameterized by tree-width of $H$.
Whenever $G$ belongs to a graph class which admits an \FPT-time FO model checking algorithm parameterized by quantifier rank, we can decide the homomorphism problem in \FPT-time parameterized by the size of $H$.
Let $V(H)=\{v_1, v_2, \ldots, v_h\}$. Then, $G$ is a model of the following formula $\phi_H$ if and only if $H$ is homomorphic to $G$:
\[\phi_H \equiv \exists x_1 \exists x_2 \ldots \exists x_h : \bigwedge\limits_{v_iv_j \in E(H)} E(x_i, x_j)\]

We use \Cref{col:algo-sc-depth} to design the following algorithm by generalizing the above formula to capture graphs $H$ of bounded tree-width:

\begin{theorem}\label{thm:algo-homomorphism}
    There exist a computable function $f: \NN^2 \to \NN$ and an algorithm $\AA$ that, given $G$ and $H$, decides whether $H$ is homomorphic to $G$ in time $\ca O(f(k, \ell) \cdot (|G|^c+|H|))$, where $k$ is the tree-width of $H$, $\ell$ is the SC-depth of $G$, and $c$ is the constant from \Cref{thm:compute-tree-model}.
\end{theorem}
\begin{proof}
    We compute a tree-decomposition $(T, (X_t)_{t \in V(T)})$ of $H$ of width $k$ in linear \FPT{} time \cite{DBLP:conf/stoc/Bodlaender93}. Using the tree-decomposition, we write a formula $\psi_H$ using $\ca O(k)$ variables, such that $G \models \psi_H$ if and only if $H$ is homomorphic to $G$.

    Choose an arbitrary root $r$ of $T$, and denote by $T_t$ the subtree of $T$ rooted at $t \in V(T)$, and denote by $Y_t$ the union of all bags in this subtree. Fix a coloring $I: V(H) \to [k+1]$ that is injective on all bags. We assume that all bags have size $k+1$.
    For each $t \in V(T)$, we write a formula $\psi_H^t(x_1, \ldots, x_{k+1})$ expressing ``there exists a homomorphic copy of $H[Y_t]$ where $x_i$ is the image of the unique vertex of color $i$ in $X_t$''.
    \[
    \psi_H^t \equiv  \bigwedge\limits_{\substack{u, v \in X_t \\ uv \in E(H)}} E(x_{I(u)}, x_{I(v)})  \land  \bigwedge\limits_{\substack{s\text{ child of }t}} \bigl[ Q_{t,s} \psi_H^s \bigr]
    \]
    where $Q_{t, s}$ is a block of existential quantifiers consisting of $\exists x_{I(w)}$ for all vertices $w \in X_s \setminus X_t$.

    We set $\psi_H \equiv \exists x_1 \ldots \exists x_{k+1} \psi_H^r$. The sentence $\psi_H$ has length $\ca O((k+1)^2\cdot|V(T)|) = \ca O_k(|H|)$ and uses only $k+1$ variables.
    We compute a tree-model of $G$ using \Cref{thm:compute-tree-model}, and we use \Cref{col:algo-sc-depth} to decide whether $G \models \psi_H$.
\end{proof}

\section{Conclusions}\label{sec:concl}

We have shown that FO model checking on colored shallow trees admits an \FPT-time algorithm parameterized by the number of variables of the formula. We have also proven that this algorithm cannot be generalized to MSO logic, even with only two set variable names, unless \FPT$=$\AW{*}.

Using the kernelization, we have derived a characterization of classes of bounded shrub-depth in purely model-theoretic terms. We have also shown that our algorithmic results can be used to design an \FPT-time algorithm for the homomorphism problem that is incomparable with the analogous algorithm derived using FO model checking metatheorems that parameterize by quantifier rank.

For monotone classes, we have fully resolved \Cref{q:fo-mc-param-by-num-var}. However, our hardness results in the hereditary setting assume that witnesses of large shrub-depth can be computed efficiently, so the question is still open there. We  conjecture that \Cref{thm:unbounded-shrub-depth-hardness} can be strengthened and that bounded shrub-depth is the right ``limit'' of efficient FO model checking on hereditary graph classes when parameterized by the number of variables.

\begin{conjecture}\label{conj:unbounded-shrub-depth-hardness}
Let $\ca C$ be a hereditary class of unbounded shrub-depth. Then, the \textsc{FO model-checking} problem is \AW{*}-hard when restricted to $\ca C$ and parameterized by the number of variables in the formula.
\end{conjecture}

\section*{Acknowledgments}
The core idea of the proof of \Cref{lem:2Pt-unbnd-shrub-depth} was discovered by ChatGPT Sol; the author has subsequently significantly simplified the proof. ChatGPT Astra noticed that MSO with two set variable names can express exact set sizes, which allowed us to prove \Cref{thm:MSO-aw-hard}. ChatGPT Sol and Astra were also used to find and correct several minor issues (typos, grammatical errors, awkward phrasing, missing corner cases, etc.) in the rest of the article. The vast majority of the article was written by the author, and the few lines edited by LLMs were subsequently reviewed and edited by the author.

\bibliography{bibliography}

\appendix

\end{document}